\documentclass[12pt]{article}
 \usepackage{a4, latexsym, amsfonts,amsmath}
 \newtheorem{theorem}{Theorem}[section]
\newtheorem{lemma}{Lemma}[section]
 \newtheorem{definition}{Definition}[section]
 \newtheorem{corollary}{Corollary}[section]

 \newtheorem{open problem}{Open problem}[section]

 \newtheorem{conjecture}{Conjecture}[section]

 \newenvironment{proof}{\trivlist
      \item[\hskip\labelsep
      {\itshape Proof.}]\normalfont}
      {\hspace*{\fill}$\Box$\endtrivlist}
\begin{document}

\title{On the Existence of Perfect Splitter Sets}
%\tnotetext[fn1]{P. Yuan's research was supported by the NSF of China (Grant No. 11271142). M. Xiong's research was supported by RGC no. 609513 from Hong Kong. }


\author{Pingzhi Yuan\thanks{P. Yuan is with School of  of Mathematical Science, South China Normal University,  Guangzhou 510631, China (email: yuanpz@scnu.edu.cn).}, Kevin Zhao \thanks{K. Zhao is with School of  of Mathematical Science, South China Normal University,  Guangzhou 510631, China (email:zhkw-hebei@163.com)}    }

\date{}
\maketitle
 \edef \tmp {\the \catcode`@}
   \catcode`@=11
   \def \@thefnmark {}

    \@footnotetext { Supported by  NSF of China  (Grant No. 11671153).}
    
    \begin{abstract} Given integers $k_1, k_2$ with $0\le k_1<k_2$, the determinations of all positive integers $q$ for which there exists a perfect Splitter $B[-k_1, k_2](q)$ set is a wide open question in general. In this paper, we obtain new  necessary and sufficient conditions for an odd prime $p$ such that there exists a nonsingular perfect $B[-1,3](p)$ set. We also give some necessary conditions for the existence of purely singular perfect splitter sets. In particular, we determine all perfect $B[-k_1, k_2](2^n)$ sets for any positive integers $k_1,k_2$ with $k_1+k_2\ge4$. We also prove that there are infinitely many prime $p$ such that there exists a perfect $B[-1,3](p)$ set.

\end{abstract}

{\bf Keywords:}
splitter set, perfect codes, factorizations of cyclic groups.

\section{Introduction} The study of splitter sets was motivated by constructing
codes correcting single limited magnitude errors used in
multilevel cell (MLC) flash memories.

Splittings were first considered in \cite{St67} in connection with the problem of tiling Euclidean space by translates of certain polytopes composed of unit cubes, called $k$-crosses and $k$-semicrosses, see also \cite{HS86} and \cite{ St84, SS94, Sz86, Sz87, SzS09}. Perfect splitter sets are equivalent to codes correcting
single limited magnitude errors in flash memories (see \cite{BE13}, \cite{EB10,KBE11, KLNY11, KLY12, M96, Sc12, Sc14, YKB13, ZG16, ZZG17, ZG18} and the references
therein). A code obtained from a perfect splitter
$B[-k_1, k_2](n)$ set can correct a symbol $a\in\{0, 1, \ldots, n-1\}$ if it is
modified into $a+e$ during transmission, where $-k_1\le e \le k_2$.

Given integers $k_1, k_2$ with $0\le k_1<k_2$, the determinations of all positive integers $q$ for which there exists a perfect Splitter $B[-k_1, k_2](q)$ set is a wide open question in general. Now there have been  many existence and nonexistence results for perfect splitter sets. In \cite{St67}, Stein showed that perfect splitter sets do not exist in some special cases, and also gave some  existence results. Kl${\o}$ve et al. \cite{KLNY11, KLY12}, gave some constructions of perfect splitter
sets for $k_1=0$ and $k_1 = k_2$. In \cite{Sc12}, Schwarz first obtained constructions of perfect splitter sets for
$1\le k_1< k_2$. For more existence and nonexistence results on perfect splitter sets, we refer to   \cite{Sc12, Sc14}, \cite{YKB13}, \cite{W95},  and \cite{ZG16, ZZG17, ZG18}. Further, Zhang and Ge \cite{ZG18} showed that there does
not exist a nonsingular perfect splitter set when  $1\le k_1 < k_2$,
 and $k_1 + k_2$ is odd.

The main purpose  of this paper is to  derive new results for perfect splitter sets.
The paper is organized as follows. In Section \ref{preliminary}, we present  some auxiliary results that will be needed in the
sequel.  We give some new results on purely singular perfect splitter sets in Section \ref{singular} and we obtain some necessary and sufficient conditions for the existence of nonsingular perfect $B[-1,3](p)$ sets in  Section \ref{nonsingular}.   Finally in Section 5 we
conclude the paper.

\section{Preliminary} \label{preliminary}

In this section, we recall some preliminary results that will be needed in the sequel.

The following notations are fixed throughout this paper.

$\bullet$  For an odd prime $p$, a primitive root $g$ modulo $p$, and an
integer $b$ not divisible by $p$, there exists a unique integer
$l\in [0, p - 2]$ such that $g^l\equiv b \pmod{ p}$.  It is known as
the index of $b$ relative to the base $g$, and it is denoted by
$ind_g(b)$.

$\bullet$ For any positive integer $q$, let $\mathbb{Z}_q$ be the ring of integers
modulo $q$ and $\mathbb{Z}_q ^\ast= \mathbb{Z}_q \backslash \{0\}$. For $a\in\mathbb{Z}_q ^\ast$, $o(a)$ denotes the order of $a$ in the multiplicative group $\mathbb{Z}_q ^\ast$.

$\bullet$  Let $a, b$ be integers such that $a\le b$, denote
$$[a, b] = \{a, a + 1, a + 2, \ldots , b\} \mbox{and}$$
$$[a, b]^\ast = \{a, a + 1, a + 2, \ldots, b\}\backslash\{0\}.$$

$\bullet$ Unless additionally defined, we assume that $aT=a\cdot T = \{a\cdot t : t\in T \}$, \,  $A+B=\{a+b, a\in A, b\in B\}$ and $AB=A\cdot B = \{a\cdot b, a\in A, b\in B\}$ for any element
$a$ and any sets $A$ and $B$, where $\cdot$ and $+$ are binary operators.

$\bullet$ For a nonempty set $M$, $|M|$ denotes the number of elements in $M$.

A. Splitting of abelian groups

Let $G$ be an abelian group, written additively, $M$ a set of
integers, and $S$ a subset of $G$. We will say that $M$ and $S$ form a {\it splitting} of
$G$ if every nonzero element $g$ of $G$ has a unique representation of the form
$g=m$s with $m\in M$ and $s \in S$, while $0$ has no such representation. (Here
"$ms$" denotes the sum of $m s$'s if $m > 0$, and  ($-( - m)s)$ if $m < 0$).  We
will write "$G \backslash \{0\} = MS$" to indicate that $M$ and $S$ form a splitting of
$G$. $M$ will be referred to as the {\it multiplier \,\, set} and S as the {\it splitting \,\, set}.

\begin{definition}\label{de21}A splitting $G \backslash \{0\} = MS$ of a finite abelian group G is
called {\it nonsingular} if every element of $M$ is relatively prime to $|G|$;
otherwise the splitting is called {\it singular}. The splitting $G \backslash \{0\} = MS$ is called
{\it purely \,\, singular} if, for every prime divisor $p$ of $| G |$ , at least one element of
$M$ is divisible by $p$. \end{definition}

We have the following two important results on splitting of abelian groups.

\begin{theorem}\label{th22} {\rm ({ \bf \cite{H83} Theorem 1.2.5.})} Let $G \backslash \{0\} = MS$ be a splitting of a finite abelian
group G. Then there exist subgroups $H$ and $K$ of $G$ such that:

(0) $G = H\times K$;

(1) the given splitting induces a nonsingular splitting of $H$;

(2) the given splitting induces a purely singular splitting of $K$.

Further, $H$ and $K$ are uniquely determined by these condition.\end{theorem}

\begin{theorem} {\rm ({\bf \cite{H83} Theorem 1.2.6.})} Let $H$ be a subgroup of the finite abelian group $G$.
Suppose $M$ splits both $H$ and $G/H$ and that the splitting of $H$ is nonsingular.
Then $M$ splits $G$.\end{theorem}

Taken together, the above two theorems reduce the study of splittings of finite
abelian groups to the study of nonsingular and of purely singular splittings. For nonsingular splittings of abelian groups,  the following theorem (\cite{HS74} Theorem 4) reduces their study to the
case of cyclic groups of prime order.

 \begin{theorem}\label{th24} {\rm ({\bf \cite{H83} Theorem 1.2.7.})} Let $G$ be a finite abelian group and $M$ a set of integers
relatively prime to $| G |$. Then $M$ splits $G$ if and only if $M$ splits $C_p$ for each
prime divisor $p$ of $|G|$.\end{theorem}

B. Splitter sets

Let $q$ be a positive integer and $k_1, k_2$ be non-negative
integers with $0\le  k_1\le  k_2$. The set $B\in\mathbb{Z}_q$ of size $n$ is
called a splitter set if all the sets
$$\{ab (mod q) : a\in [-k_1, k_2]\},\quad  b\in B $$
have $k_1+k_2$ nonzero elements, and they are disjoint. We denote
such a splitter set by $B[-k_1, k_2](q)$ set.

If a $B[-k_1, k_2](q)$ set of size $n$ exists, then we have
$$q\ge (k_1 + k_2)n + 1,$$
and so
$$n\le\frac{q- 1}{k_1 + k_2}.$$

A $B[-k_1, k_2](q)$ set is called perfect if $n = \frac{q- 1}{k_1 + k_2}$. Clearly,
a perfect set can exist only if $q\equiv1\pmod{ k_1 + k_2}$.   From the definition of perfect $B[-k_1, k_2](q)$ sets, we have

\begin{theorem}\label{th25} Let $q$ be a positive integer and $k_1, k_2$ be non-negative
integers with $0\le  k_1\le  k_2$. Let $M=[-k_1, k_2]^\star$. Then $B$ is a perfect $B[-k_1, k_2](q)$ set if and only if $MB$ is a splitting of $\mathbb{Z}_q$. \end{theorem}
By Theorem \ref{th25} and Definition \ref{de21}, we say that a perfect $B[-k_1, k_2](q)$ set
is {\it nonsingular} if $\gcd(q, k_2!) = 1$. Otherwise, the set is called
{\it singular}. If for any prime $p|q$, there is some $k$ with $0 < k \le k_2$
such that $p|k$, then the perfect $B[-k_1, k_2](q)$ set is called {\it
purely \, singular}.

{\bf Remark:} By Theorems \ref{th22} and \ref{th25}, we need only consider purely singular perfect $B[-k_1, k_2](q)$ sets and nonsingular perfect $B[-k_1, k_2](q)$ sets.  By Theorems \ref{th24} and \ref{th25},  there
is a  nonsingular perfect $B[-k_1, k_2](q)$ set if and only if there
is a  nonsingular  perfect $B[-k_1, k_2](p)$ set for each prime $p$
dividing $q$. Therefore, we are only interested in considering  purely singular perfect $B[-k_1, k_2](q)$ sets and nonsingular  perfect $B[-k_1, k_2](p)$ sets for an odd prime $p$.

C. Group factorizations

\begin{definition}Let $(G, +)$ be an abelian group. If each element $g\in G$ can be expressed uniquely in the form
$$g = a + b, a \in A,\,\, b \in B,$$
then the equation $G = A + B$ is called a {\it factorization} of $G$. A non-empty subset of $G$ is called to be a {\it direct  factor} of $G$ if there exists a subset $B$ such that $G=A+B$ is a factorization. \end{definition}

The following elementary lemma is useful for the existence of a factorization of finite abelian groups.

\begin{lemma}\label{le27} Let $G$ be a finite abelian group and let $A, B$ be non-empty subsets of $G$. The
following statements are all equivalent to the fact that the equation $G = A+B$
is a factorization of $G$.

(i) The sum $A + B$ is direct and is equal to $G$.

(ii) $G = A + B$ and $|G| = |A||B|$.

(iii) $|G| = |A||B|$ and $(A - A) \cap (B - B) \subseteq \{0\}$.

(iv) $G = A + B$ and $(A - A) \cap (B - B) \subseteq \{0\}$.

(v) The sets $A + b, b \in B$ form a partition of $G$.

(vi) The sets $a + B, a \in A$ form a partition of $G$.\end{lemma}

For an odd prime $p$ and a primitive root $g$ modulo
$p$,  we let
$$N=\{ind_g(k)| k\in[-k_1, k_2]^\ast\}, \quad A=\{ind_g(b) |b\in B\},$$
where $B$ is a non-empty subset of $\mathbb{Z}_p\backslash\{0\}$.
Then we have the following obviously lemma.

\begin{lemma}\label{le28} $BM=\mathbb{Z}_p^\ast$ if and only if
$$N+A=\mathbb{Z}_{(p-1)}.$$
\end{lemma}
Therefore, we have
 \begin{theorem}\label{th29} $B$ is a nonsingular perfect $B[-k_1, k_2](p)$ set for a prime $p$ if and only if $N+A=\mathbb{Z}_{(p-1)}$ is a factorization, and there exists a nonsingular perfect $B[-k_1, k_2](p)$ set if and only if $A$ is a  direct factor of $\mathbb{Z}_{(p-1)}$.\end{theorem}

 By Lemmas \ref{le27}, \ref{le28} and Theorem \ref{th29}, we have the following lemma which is a powerful tool to
derive  necessary conditions for the existence of nonsingular perfect $B[-k_1, k_2](p)$ sets.
\begin{lemma}\label{lemain} Let $p$ be a prime and $g$ a primitive root modulo $p$, let $k_1, k_2$ be integers such that
$1\le k_1\le k_2$. Set $N=\{ind_g(k)| k\in[-k_1, k_2]^\ast\}, \quad A=\{ind_g(b) |b\in B\}$. Then $B$ is a nonsingular perfect
$B[-k_1, k_2](p)$ set if and only if one of the following  conditions holds:

(i) $N+A=\mathbb{Z}_{(p-1)}$ is a factorization;

(ii) $N$ is a direct factor of $\mathbb{Z}_{(p-1)}$;

(iii) $(k_1+k_2)|B|=p-1$ and $(N-N)\cap(A-A)\subseteq\{0\}$;

(iv) $(k_1+k_2)|B|=p-1$ and $(N-N)\cap(A-A)\subseteq\{0\}$;

(v) The set $N+a, a\in A$ form a partition of $\mathbb{Z}_{(p-1)}$;

(vi) The set $A+n, n\in N$ form a partition of $\mathbb{Z}_{(p-1)}$;

(vii) For any $a\in\mathbb{Z}_p^\ast$, $|B \cap a[-k_1, k_2]^\ast| = 1$.\end{lemma}

D. $k$th power residue modulo $m$

Let $m, k$ and $a$ be integers such that $\gcd(m, a)=1$. we say that $a$ is a $k$th power residue modulo $m$ if there exists an integer $x$ such that
$$x^k\equiv a\pmod{m}.$$
If this congruence has no solution, then $a$ is called a $k$th power nonresidue modulo $m$. We have the following well-known result for the $k$th power residue modulo a prime $p$, which will be used in this paper.

\begin{theorem} Let $p$ be an odd prime, $d\ge2$ and $d|p-1$. Let $a$ be an integer not divisible by $p$. Let $g$ be a primitive root modulo $p$. Then $a$ is a $k$th power residue modulo $p$ if and only if
$$ind_g(a)\equiv 0\pmod{d}$$
if and only if
$$a^{(p-1)/d}\equiv 1\pmod{p}.$$
\end{theorem}

\section{Purely Singular Perfect Splitter Sets} \label{singular}

We need the following results on the factorization of cyclic groups.

\begin{theorem}\label{Sath} {\rm ({\bf \cite{Sa57} Theorem 1.4})} Let $G$ be a finite cyclic group (written multiplicatively) and let $ G = AB$ be a factoring of $G$. Assume that $1\in A\cap B$. (1 is the identity of $G$.) Suppose also that the order of $A$ is a power of a prime. Then at least one of
the sets $A$ and $B$ consists of cosets of a subgroup of $G$ of order greater than one. \end{theorem}

The following lemma is a special case of [8, Theorem 1.2.1] %\cite{H83}
\begin{lemma}\label{Lemma} If $m|n$ and there exist both a perfect
$B[-k_1, k_2](m)$ set and a perfect $B[-k_1, k_2](n)$ set, then there
exists a perfect $B[-k_1, k_2](n/m)$ set.\end{lemma}

The following result of \cite{ZG18} ([31, Lemma 11]),  is a generalization of  [6, Theorem 2.1] and %\cite{GS81}
the proof is similar. However, the proof in \cite{ZG18} is incomplete. For the sake of completeness, we give the
proof here.
\begin{theorem}\label{Thzg}  Let $k_1, k_2$ be integers, $1\le k_1 \le k_2, k_2\ge 3$,
$n = k_1 + k_2 + 1$. If $n$ is not a prime, then there does not exist
a perfect $B[-k_1, k_2](n^2)$ set.\end{theorem}

\begin{proof}We follow the argument of Galovich and Stein \cite{GS81}.

Let $G = \mathbb{Z}_{n^2}$. The proof that $[-k_1, k_2]^\ast$ does not split $G$ is
divided into three cases: (i) $n$ is not a prime power, (ii) $n$ is a power of $2$, (iii) $n$ is
a power of an odd prime. We will consider each case after some preliminary
observations.

Assume that $\mathbb{Z}_{n^2}\backslash\{0\}=[-k_1, k_2]^\ast S$. The number of elements in $S$
relatively prime to $n$ is $\varphi(n^2)/\varphi(n) = n$. Thus $S =\{x, a_1, \ldots, a_n\}$ where $\gcd(x, n)>1$
and $(a_i, n)=1$ for $1\le i\le n$. If an element $jn\in \mathbb{Z}_{n^2}\backslash\{0\}, 1\le j\le n-1$, were of the
form $ia_k$, then $jn\equiv ia_k \pmod{ n^2}$; thus $n$ divides $i$, contradicting the fact that
$i\in [-k_1, k_2]^\ast$. Consequently the $n - 1$ elements $n, 2n, \ldots, (n - 1)n$ are a permutation of the elements $x, 2x, \ldots, (n- 1)x\pmod{n^2}$. Thus $x = in$ for some integer
$1 \le j\le n-1$. Since there is an integer $j, j\in [-k_1, k_2]^\ast$, such that $jin\equiv n \pmod {n^2}$ or
equivalently $ji\equiv1 \pmod {n}$, $\gcd(i, n) = 1$, and one can assume that $x = n$. Henceforth
it will be assumed that $S = \{n, a_1,\ldots, a_n\}$.

(i) Let $n = pqm$ where $p$ and $q$ are distinct primes, $p <q$, and $m\in\mathbb{N}$. Let
$d = q^2m$. Observe that $n < d < n^2$, $d$ divides $n^2$ while $n$ does not divide $d$.

We claim that $d$ is not represented in the alleged splitting. Note first that $d\not\equiv in\pmod {n^2}$. Thus assume that $d\equiv ia_j\pmod {n^2}$ for some $i \in [-k_1, k_2]^\ast$. Let $k$ be the integer $n^2/d$. Clearly $k < n$. Then we have $0\equiv kd \equiv kia_j \pmod {n^2}$. Since $\gcd(a_j, n) = 1$, $ki\equiv0\pmod {n^2}$ . But since $1\le|i|, |k|\le n-1$, $1\le|ik|\le(n-1)^2$ showing that $ik\not\equiv0\pmod {n^2}$. This contradiction completes the proof in case (i).

(ii) Let $k$ be an integer $\ge2$ and let $k_1+k_2=2^k-1, \,\, M=[-k_1, k_2]^\ast$. Assume that
$M\{2^k, a_1, \ldots, a_{2^k}\}= \mathbb{Z}_{2^{2k}}\backslash\{0\}$. Then $M\{a_1, \ldots,  a_{2^k}\} = P$,  where $P =\mathbb{Z}_{2^{2k}}\backslash2^k\{0, 1, \ldots, 2^k-1\}$, the elements in $\mathbb{Z}_{2^{2k}}$ that are not multiples of $2^k$.

Let $A = \{a_1, \ldots, a_{2^k}\}$. We shall examine within $P$ the multiples of $2^{k-1}, 2^{k-2}$ and $2^{k-3}$.

Let $p = t2^{ k-1}\in P$. Then $t$ is odd. Note that there are $\varphi(2^{k+1}) = 2^k$ choices of $t$, and all such $t$ form a reduced set of residues modulo $2^{k+1}$.
Write $p = ma$ where $m\in M$ and $a\in A$. Since $a$ is odd, $m$ is a multiple of $2^{k-1}$;
since $-2^{k-1}<m <2^k, m = 2^{ k-1}$. Thus, for $t$ odd, there exists $a\in A$ such that $t\cdot2^{ k-1}\equiv2^{k-1}\cdot a\pmod{2^{2k}}$ or
$t\equiv a\pmod{2^{k+1}}$. Since there are exactly $2^k$ choices of $t$ and $2^k$
elements in $A$, so $A$ is a reduced set of residues modulo $2^{k+1}$.

Next let $q = u2^{k-2}\in P$ where $u$ is odd. Writing $u2^{ k-2}= ma$ where $m\in M$ and
$a\in A$, one concludes that $m$ is an odd multiple of $2^{ k-2}$, hence either $m = 2^{ k-2}, -2^{ k-2}$ or
$m=3\cdot2^{ k-2}$ as $k_1\le k_2$ and $5\cdot2^{ k-2}>2^k$. We divide the proof into two cases.

{\bf Case 1:} $k_2\ge 3\cdot2^{ k-2}$. In this case, $m = 2^{ k-2}$ or
$m=3\cdot2^{ k-2}$. Let $a_0$ be an arbitrary element of $A$. Consider in $P$ the element $9\cdot2^{ k-2}\pmod {2^{2k}}$. This element has
the form $2^{ k-2}a_1$ or $3\cdot2^{ k-2}a_1$ for some $a_1\in A$. In the
second case $9\cdot2^{ k-2}a_0\equiv3\cdot2^{ k-2}a_1 \pmod {2^{2k}}$, hence $3a_0\equiv a_1\pmod {2^{k+2}}$. Thus the
element $2^{ k-2}a_1$ has two representations in the form $ma$, namely $2^{ k-2}a_1$ and
$(3\cdot2^{ k-2})a_0$. A contradiction.

The first case, $9\cdot2^{ k-2}a_0\equiv2^{ k-2}a_l\pmod {2^{2k}}$, implies that $9a_0\equiv a_1\pmod {2^{k+2}}$.
We repeat the argument with $a_1$ in place of $a_0$. If the second case does not occur,
then the argument may be repeated again and continued. Assuming that the
second case does not occur, we have for each positive integer $r$ an element $a_r\in A$ such that $9^ra_0\equiv a_r\pmod{ 2^{k+2}}$.

Now $9\equiv2^{2+1} + 1\pmod{ 2^4}$ and by induction on $k$, $9^{2^{k-2}}\equiv2^{k+1} + 1\pmod {2^{k+2}}$.
Thus for $r = 2^{ k-2}$, $9^ra_0\equiv a_r\pmod {2^{ k+2}}$ and also $9^{2^{k-2}}\equiv2^{k+1} + 1\pmod {2^{k+2}}$, hence
$9^ra_0\equiv a_0\pmod {2^{ k+1}}$. Thus $a_0\equiv a_r\pmod {2^{ k+1}}$. Since $a_0, a_r\in A$ and $A$ is a reduced set of residues modulo $2^{k+1}$, so $a_0=a_r$, which implies that  $9^ra_0\equiv a_0\pmod {2^{ k+2}}$ and $9^{2^{k-2}}\equiv2^{k+1} + 1\pmod {2^{k+2}}$.
This contradiction completes the argument for the Case 1.

{\bf Case 2:} $k_2< 3\cdot2^{ k-2}$. In this case, $m = 2^{ k-2}$ or
$m=-2^{ k-2}$.

Let $p = t2^{ k-2}\in P$ and $t$ is odd. Note that there are $\varphi(2^{k+2}) = 2^{k+1}$ choices of $t$, and all such $t$ form a reduced set of residues modulo $2^{k+2}$..
Write $p = ma$ where $m\in M$ and $a\in A$. Since $a$ is odd, $m$ is an odd multiple of $2^{k-2}$;
since $-2^{k-1}<m <2^k, m = 2^{ k-2}$ or $-2^{k-2}$. Thus, for $t$ odd, there exists $a\in A$ such that $t\cdot2^{ k-2}\equiv\pm2^{k-2}\cdot a\pmod{2^{2k}}$ or
$t\equiv \pm a\pmod{2^{k+2}}$. Since there are exactly $2^{k+1}$ choices of $t$ and $2^k$
elements in $A$, moreover, $\{t, t2^{k-2}\in P, 2\not|t\}$ is a reduced set of residues modulo $2^{k+2}$,  so $\pm A$ is a reduced set of residues modulo $2^{k+2}$.

Let $q = u\cdot2^{ k-3}\in P$ and $u$ is odd. Writing $u\cdot2^{ k-3}=ma$ where $m\in M$ and $a\in A$. Since $a$ is odd, $m$ is an odd multiple of $2^{k-3}$; hence either $m=-3\cdot2^{k-3}$ or  $-2^{k-3}$ or $2^{k-3}$ or $3\cdot2^{k-3}$ or $5\cdot2^{k-3}$ as $k_2<6\cdot2^{k-3}$. We divide the remaining proof of this case into two subcases.

{\bf Subcase 2.1:} $k_2\ge5\cdot2^{k-3}$. Then $k_1<3\cdot2^{k-3}$.  Let $a_0$ be an arbitrary element of $A$. Consider in $P$ the element $15\cdot2^{ k-3}\pmod {2^{2k}}$. This element has
the form $2^{ k-3}a_1$ or $-2^{ k-3}a_1$ or $3\cdot2^{ k-3}a_1$ or $5\cdot2^{ k-3}a_1$ for some $a_1\in A$. In the
third case $15\cdot2^{ k-3}a_0\equiv3\cdot2^{ k-3}a_1 \pmod {2^{2k}}$, hence $5a_0\equiv a_1\pmod {2^{k+3}}$. Thus the
element $2^{ k-3}a_1$ has two representations in the form $ma$, namely $2^{ k-3}a_1$ and
$(5\cdot2^{ k-2})a_0$.  In the
fourth case $15\cdot2^{ k-3}a_0\equiv5\cdot2^{ k-3}a_1 \pmod {2^{2k}}$, hence $3a_0\equiv a_1\pmod {2^{k+3}}$. Thus the
element $2^{ k-3}a_1$ has two representations in the form $ma$, namely $2^{ k-3}a_1$ and
$(3\cdot2^{ k-2})a_0$. A contradiction.

In the
second case $15\cdot2^{ k-3}a_0\equiv-2^{ k-3}a_1 \pmod {2^{2k}}$, hence $-15a_0\equiv a_1\pmod {2^{k+3}}$. In the first case $15\cdot2^{ k-3}a_0\equiv2^{ k-3}a_1 \pmod {2^{2k}}$, hence $15a_0\equiv a_1\pmod {2^{k+3}}$.  We repeat the argument with $a_1$ in place of $a_0$. If the third and fourth cases do not occur,
then the argument may be repeated again and continued. Assuming that the
third and fourth cases do not occur, we have for each positive integer $r$ an element $a_r\in A$ such that $\pm 15^ra_0\equiv a_r\pmod{ 2^{k+3}}$.

Now $15\equiv 2^4-1\pmod{2^5}$ and by induction on $k$, we have
$$15^{2^{k-2}}\equiv 2^{k+2}+1\pmod{2^{k+3}}, k\ge3.$$
Thus for $r=2^{k-2}$, $\pm 15^ra_0\equiv a_r\pmod{2^{k+3}}$ and also $15^{2^{k-2}}\equiv 2^{k+2}+1\pmod{2^{k+3}}$, hence $15^ra_0\equiv a_0\pmod{2^{k+2}}$. Thus $a_0\equiv \pm a_r\pmod{2^{k+2}}$. Since $\pm A$ is a reduced set of residues modulo $2^{k+2}$, we have $a_r=a_0$. If $a_0=a_r$, then $\pm15^ra_0\equiv a_0\pmod{2^{k+3}}$ and $\pm 15^r\equiv\pm1\equiv\pm(
2^{k+2}+1)\pmod{2^{k+3}}$. This contradiction completes the argument in this subcase.

{\bf Subcase 2.2:} $k_2<5\cdot2^{k-3}$. Then $k_1\ge3\cdot2^{k-3}$.  Let $a_0$ be an arbitrary element of $A$. Consider in $P$ the element $9\cdot2^{ k-3}\pmod {2^{2k}}$. This element has
the form $2^{ k-3}a_1$ or $-2^{ k-3}a_1$ or $3\cdot2^{ k-3}a_1$ or $-3\cdot2^{ k-3}a_1$ for some $a_1\in A$. In the
third case $9\cdot2^{ k-3}a_0\equiv3\cdot2^{ k-3}a_1 \pmod {2^{2k}}$, hence $3a_0\equiv a_1\pmod {2^{k+3}}$. Thus the
element $2^{ k-3}a_1$ has two representations in the form $ma$, namely $2^{ k-3}a_1$ and
$(3\cdot2^{ k-2})a_0$.  In the
fourth case $9\cdot2^{ k-3}a_0\equiv-3\cdot2^{ k-3}a_1 \pmod {2^{2k}}$, hence $-3a_0\equiv a_1\pmod {2^{k+3}}$. Thus the
element $2^{ k-3}a_1$ has two representations in the form $ma$, namely $2^{ k-3}a_1$ and
$(-3\cdot2^{ k-2})a_0$. A contradiction.

In the
second case $9\cdot2^{ k-3}a_0\equiv-2^{ k-3}a_1 \pmod {2^{2k}}$, hence $-9a_0\equiv a_1\pmod {2^{k+3}}$. In the first case $9\cdot2^{ k-3}a_0\equiv2^{ k-3}a_1 \pmod {2^{2k}}$, hence $9a_0\equiv a_1\pmod {2^{k+3}}$.  We repeat the argument with $a_1$ in place of $a_0$. If the third and fourth cases do not occur,
then the argument may be repeated again and continued. Assuming that the
third and fourth cases do not occur, we have for each positive integer $r$ an element $a_r\in A$ such that $\pm 9^ra_0\equiv a_r\pmod{ 2^{k+3}}$.

Now $9\equiv 2^3+1\pmod{2^4}$ and by induction on $k$, we have
$$9^{2^{k-1}}\equiv 2^{k+2}+1\pmod{2^{k+3}}, k\ge3.$$
Thus for $r=2^{k-1}$, $\pm 9^ra_0\equiv a_r\pmod{2^{k+3}}$ and also $9^{2^{k-2}}\equiv 2^{k+2}+1\pmod{2^{k+3}}$, hence $9^ra_0\equiv a_0\pmod{2^{k+2}}$. Thus $a_0\equiv \pm a_r\pmod{2^{k+2}}$. Since $\pm A$ is a reduced set of residues modulo $2^{k+2}$, we have $a_r=a_0$. If $a_0=a_r$, then $\pm9^ra_0\equiv a_0\pmod{2^{k+3}}$ and $\pm 9^r\equiv\pm1\equiv\pm(
2^{k+2}+1)\pmod{2^{k+3}}$. This contradiction completes the argument in this subcase. This completes the proof of case (ii).

(iii) Let $p$ be an odd prime and let $k\ge 2$ be an integer. Assume that $k_1+k_2=p^k-1$, $1\le k_1\le k_2$ and
$$\{-k_1, -k_1+1, \ldots, -1, 1, \ldots, k_2\}\{p^k,  a_1, \ldots, a_{p^k}\} = \mathbb{Z}_{p^{2k}}\backslash\{0\}.$$
For a subset $X\subseteq\mathbb{Z}_{p^{2k}}$, let $X^\star =\{x\in X |\gcd(x, p) = 1\}$. Then
$$\{-k_1, -k_1+1, \ldots, -1, 1, \ldots, k_2\}^\star\{  a_1, \ldots, a_{p^k}\}$$
is a factorization of the group $\mathbb{Z}_{p^{2k}}^\star$. Since $\mathbb{Z}_{p^{2k}}^\star$ is a cyclic group, Theorem \ref{Sath} implies that one of the factors consists of cosets of a cyclic subgroup. But for such a factor, call it $B$, there is an element $g\ne 1$ such that $gB = B$. This condition is not satisfied by the set $\{-k_1, -k_1+1, \ldots, -1, 1, \ldots, k_2\}^\star$ since, if $|g^a|\le k_2$ and $|g^{a+1}|\ge k_2+1$, then $g^{a+1}\not\in\{-k_1, -k_1+1, \ldots, -1, 1, \ldots, k_2\}^\star$ since $k_2^2>k_2$ and $k_2^2+k_1\le(p^k-1)^2<p^{2k}$. Thus there is an element $h\ne1$ and a set $C$ such that
$\{a_1, \ldots, a_{p^k}\} = (h)C$ where $(h)$ is the group generated by $h$. It is no loss of generality to assume that $(h)$ has $p$ elements. Thus $(h) =\{y\in \mathbb{Z}_{p^{2k}}^\star|
y\equiv1\pmod{ p^{2k}}$.
Now for any $c\in C$, both $c$ and $(1+p^{2k-1})c$ are elements of $\{a_1, \ldots, a_{p^k}\}$.
Hence $pc\equiv p(1+p^{2k-1})c\pmod{p^{ 2k}}$, a contradiction to the factorization of $\mathbb{Z}_{p^{2k}}^\star$. This completes the proof of the theorem.

\end{proof}

The following lemma is a generalization of  Lemma 12 in \cite{ZG18}.

\begin{lemma} Suppose there exists a perfect $B[-k_1, k_2](m)$ set.
Suppose also there exist a prime $p$ and an integer $a > 0$  such
that $p|m$ and $a|p- 1$. Let $r$ be a positive integer with $\gcd(a(k_1+k_2), r)=1$ and $p|a(k_1+k_2)+r$. If $r\le a$ and  $\lfloor\frac{k_1}{p}\rfloor+\lfloor\frac{k_2}{p}\rfloor=\lfloor\frac{k_1+k_2}{p}\rfloor$. Then $a(k_1 + k_2) + r|m$.\end{lemma}
\begin{proof} Let $B = \{s_1,\ldots,s_n\}$ be a perfect $B[-k_1, k_2](m)$ set,
and suppose $p|s_i$ for $1\le i\le t$ and $p\not|s_i$ for $t + 1 \le i \le n$.

Let $a$ be a prime divisor of $p-1$. Since $p|a(k_1+k_2)+r$, so there is a positive integer $v$ such that
$$a(k_1+k_2)=vp-r, \quad \quad v\in\mathbb{N}.$$
Let $v=aq+s, 0\le s<a$. Then  $a(k_1+k_2)=aqp+s(p-1)+s-r$. Since $a|p-1$, so $a|s-r$, it follows from $r\le a$ that  $s=r$, then
$$a(k_1+k_2)=apq+r(p-1), \quad q=\lfloor\frac{k_1+k_2}{p}\rfloor.$$
Note that $|\{i, p|i, i\in[-k_1, k_2]^\ast\}|=\lfloor\frac{k_1}{p}\rfloor+\lfloor\frac{k_2}{p}\rfloor=\lfloor\frac{k_1+k_2}{p}\rfloor$ by the assumptions, so
$$\frac{(k_1+k_2)n+1}{p}=|<p>|=1+(k_1+k_2)t+\left(\lfloor\frac{k_1}{p}\rfloor+\lfloor\frac{k_2}{p}\rfloor\right)(n-t)$$
$$=1+(k_1+k_2)t+\left(\lfloor\frac{k_1+k_2}{p}\rfloor\right)(n-t)$$
$$=1+(k_1+k_2)t+\frac{a(k_1+k_2)-r(p-1)}{ap}(n-t),$$
which implies that
$$n=at(k_1+k_2)/r+a/r+t.$$
Therefore
$$rm=(r(k_1+k_2)n+1)=\left(a(k_1+k_2)+r\right)((k_1+k_2)t+1).$$
Since $\gcd(a(k_1+k_2)+r, r)=1$, so $a(k_1 + k_2) + r|m$. The lemma is proved.

\end{proof}
Let $r=1$, we obtain
\begin{corollary}\label{Cozg}Let $k_1, k_2$ be positive integers with $k_1+k_2\ge4$. Suppose there exists a perfect $B[-k_1, k_2](m)$ set.
Suppose also there exist a prime $p$ and an integer $a > 0$  such
that $p|m$ and $a|p- 1$. If $p|a(k_1+k_2)+1$ and  $\lfloor\frac{k_1}{p}\rfloor+\lfloor\frac{k_2}{p}\rfloor=\lfloor\frac{k_1+k_2}{p}\rfloor$. Then $a(k_1 + k_2) + 1|m$.\end{corollary}

{\bf Remark:} Corollary \ref{Cozg} tells us that we need  an additional condition for  Lemma 13 in \cite{ZG18}.

We also have

\begin{lemma}\label{Lezg}Let $k_1, k_2$ be positive integers with $k_1+k_2\ge4$. Suppose there exists a perfect $B[-k_1, k_2](m)$ set with $k_1+k_2+1$ composite, then either

$\bullet$ $\gcd(k_1+k_2+1, m)=1$, or

$\bullet$ $k_1+k_2+1|m$ and $\gcd(k_1+k_2+1, \frac{m}{k_1+k_2+1})=1$.\end{lemma}

\begin{proof} Assume $\gcd(k_1+k_2+1, m) > 1$. Applying Corollary \ref{Cozg}
with $a = 1$ and $p$ being any prime divisor of $\gcd(k_1+k_2+1, m)$, it is easy to check that $\lfloor\frac{k_1}{p}\rfloor+\lfloor\frac{k_2}{p}\rfloor=\lfloor\frac{k_1+k_2}{p}\rfloor$ since $p|k_1+k_2+1$, so we obtain that $k_1+k_2+1|m$. Since there exist both a
perfect $B[-k_1, k_2](m)$ set and a perfect $B[-k_1, k_2](k_1+k_2+1)$ set, then there exists a perfect $B[-k_1, k_2](m/k_1+k_2+1)$ set by Lemma \ref{Lemma}. If $\gcd
(k_1+k_2+1, m/k_1+k_2+1) > 1$, we can repeat the above argument and get a perfect $B[-k_1, k_2](m/(k_1+k_2+1)^2)$ set. Then by Lemma \ref{Lemma},
we have a perfect $B[-k_1, k_2]((k_1+k_2+1)^2)$ set, which contradicts Theorem \ref{Thzg}.\end{proof}

By applying the above results, we have the following result, which is one of the main results in this paper.
\begin{theorem}\label{ma1} Let $k_1, k_2, n$ be positive integers with $k_1+k_2\ge4$. Suppose there exists a purely singular perfect $B[-k_1, k_2](2^n)$ set. Then $2^n=k_1+k_2+1$.\end{theorem}

\begin{proof}Suppose there exists a purely singular perfect $B[-k_1, k_2](2^n)$ set. Then $k_1+k_2|2^n-1$, which implies that $k_1+k_2$ is odd, and $2|\gcd(k_1+k_2+1, 2^n)$. Applying Corollary \ref{Cozg} with $a=1$ and $p=2$, we see that $k_1+k_2+1|2^n$. Now applying Lemma \ref{Lezg} with $m=2^n$ and $k_1+k_2+1|2^n$, we obtain that $\gcd(k_1+k_2+1, \frac{2^n}{k_1+k_2+1})=1$, which implies that $2^n=k_1+k_2+1$.\end{proof}
Note that Schwartz \cite{Sc12} has constructed an infinite family of
purely singular perfect $B[-1, 2](4^l)$ sets, so the  restriction $k_1+k_2\ge4$ is indispensable.
More general, we have
\begin{theorem}\label{ma2} Let $k_1, k_2, n$ be positive integers with $k_1+k_2\ge4$ and let $p$ be a prime with $p|k_1+k_2+1$ and $p\ne k_1+k_2+1$. Suppose there exists a perfect $B[-k_1, k_2](p^n)$ set. Then $p^n=k_1+k_2+1$.\end{theorem}

\begin{proof}Suppose there exists a perfect $B[-k_1, k_2](p^n)$ set. Since $p|\gcd(k_1+k_2+1, p^n)$. Applying Corollary \ref{Cozg} with $a=1$ and $p=p$, we see that $k_1+k_2+1|p^n$. Now applying Lemma \ref{Lezg} with $m=p^n$ and $k_1+k_2+1|p^n$, we obtain that $\gcd(k_1+k_2+1, \frac{p^n}{k_1+k_2+1})=1$, which implies that $p^n=k_1+k_2+1$. This proves the theorem.\end{proof}

Theorems \ref{ma1} and \ref{ma2} give a partial answer for the following conjecture proposed by Zhang and Ge \cite{ZG18}.

\begin{conjecture} Let $k_1, k_2$ be integers with $1 \le k_1 < k_2$ and
$k_1 + k_2\ge 4$, then there does not exist any purely singular
perfect $B[-k_1, k_2](m)$ set except for $m = 1$ and except
possibly for $m = k_1 + k_2 + 1$.\end{conjecture}

\section{Nonsingular Perfect Splitter Sets} \label{nonsingular}
In this section we will prove  new existence results for
nonsingular perfect $B[-1, 3]^\ast(p)$ sets. We first prove the following general result.

\begin{theorem}Let $k_1, k_2$ be positive integers with $1\le k_1\le k_2$ and let $p$ be an odd prime with $p\equiv1\pmod{k_1+k_2}$. Then $M=[-k_1, k_2]^\ast$ is a direct factor of $\mathbb{Z}_p^\ast$ if and only if $M$ is a direct factor of the subgroup $H=<-1, 2, \ldots, k_2>$ of $\mathbb{Z}_p^\ast$.\end{theorem}
\begin{proof}If $M=[-k_1, k_2]^\ast$ is a direct factor of $\mathbb{Z}_p^\ast$, then there exists a subset $B\subseteq \mathbb{Z}_p^\ast$ such that $MB=\mathbb{Z}_p^\ast$. Let $B_1=B\cap H$. Then it is easy to see that $mb\in H, m\in M, b\in B$ if and only if $b\in H$, so $MB_1=H$. Obviously, it is a factorization of $H$.

Now if $M$ is a direct factor of the subgroup $H=<-1, 2, \ldots, k_2>$ and $H=MB_1$ is a factorization. Let $|\mathbb{Z}_p^\ast/H|=t$ and $\mathbb{Z}_p^\ast=\uplus_{i=1}^tg_iH$, and let $B=\uplus_{i=1}^t B_1g_i$. It is easy to check that $\mathbb{Z}_p^\ast=\uplus_{i=1}^tg_iH=\uplus_{i=1}^tMB_1g_i=MB$ is a factorization of $\mathbb{Z}_p^\ast$. This completes the proof. \end{proof}

{\bf Remark:} Similarly, we can prove that: let  $p$ be an odd prime and $M$ is a nonempty subset of $\mathbb{Z}_p^\ast$ with $p\equiv1\pmod{|M|}$. Then $M$ is a direct factor of $\mathbb{Z}_p^\ast$ if and only if $M$ is a direct factor of the subgroup $H=<M>$ of $\mathbb{Z}_p^\ast$, where $<M>$ denotes the subgroup generated by the set $M$.

We also need the following result for the factorization of cyclic groups.
\begin{theorem}\label{thn0} (\cite{SzS09} Theorem 7.1) Let $m$ and $n$ be relatively prime positive integers. If $A = \{a_1, \ldots , a_m\}$ and
$B = \{b_1, \ldots , b_n\}$ are sets of integers such that their sum set
$$A + B = \{a_i + b_j : 1 \le i \le m, 1 \le j \le n\}$$
is a complete set of representatives modulo $mn$, then $A$ is a complete set of
residues modulo $m$ and $B$ is a complete set of residues modulo $n$.\end{theorem}

\begin{theorem}\label{thn1}Let $p$ be an odd prime with $p\equiv1\pmod{4}$, and $B$  a perfect $B[-1, 3]^\ast(p)$ set for  $p$. If $i\in B$, then
$$i<-\frac{3}{2}>\in B,$$
where $<-\frac{3}{2}>$ denotes the subgroup of $\mathbb{Z}_p^\ast$ generated by $-\frac{3}{2}$.\end{theorem}
\begin{proof}Let $M=\{-1, 1, 2, 3\}$. Since $B$  a perfect $B[-1, 3]^\ast(p)$ set, by Lemma \ref{lemain} (vii), for any $a\in\mathbb{Z}_p^\ast$, $|B\cap aM|=1$.

Taking $a=i$, we have
$$aM=\{-i, i, 2i, 3i\},$$
so $-i, 2i, 3i\not\in B$.

Taking $a=-i$, we have
$$aM=\{-i, i, -2i, -3i\},$$
so $-2i, -3i\not\in B$.

Taking $a=i/2$, we have
$$aM=\{-i/2, i/2, i, 3i/2\},$$
so $-i/2, i/2, 3i/2\not\in B$.

Taking $a=-i/2$, we have
$$aM=\{-i/2, i/2, -i, -3i/2\},$$
so $-3i/2\in B$.

Taking $a=i/3$, we have
$$aM=\{-i/3, i/3, 2i/3, i\},$$
so $-i/3, i/3, 2i/3\not\in B$.

Taking $a=-i/3$, we have
$$aM=\{-i/3, i/3, -2i/3, -i\},$$
so $-2i/3\in B$.

Therefore we have shown that for any $i\in B$, we have $-3i/2, -2i/3\in B$. Hence $i<-\frac{3}{2}>\in B$, where $<-\frac{3}{2}>$ denotes the subgroup of $\mathbb{Z}_p^\ast$ generated by $-\frac{3}{2}$. This proves the theorem. \end{proof}

\begin{lemma}Let $p$ be an odd prime with $p\equiv1\pmod{4}$, and $B$  a perfect $B[-1, 3]^\ast(p)$ set for  $p$. If $i\in B$, then
$6i\in B$ or $-6i\in B$.\end{lemma}
\begin{proof}
Let $M=\{-1, 1, 2, 3\}$. Since $B$  a perfect $B[-1, 3]^\ast(p)$ set, by Lemma  \ref{lemain} (vii), for any $a\in\mathbb{Z}_p^\ast$, $|B\cap aM|=1$.

Taking $a=2i$, we have
$$aM=\{-2i, 2i, 4i, 6i\},$$ note that $2i, -2i\not\in B$, so $4i\in B$ or $6i\in B$.

Taking $a=-2i$, we have
$$aM=\{-2i, 2i, -4i, -6i\},$$
similarly, we have $-4i\in B$ or $-6i\in B$.

If $4i\in B$, then  $-4i\not\in B$, so $-6i\in B$. If $4i\not\in B$, then $6i\in B$. This proves the lemma.\end{proof}

Let $o(a)$ denote the order of $a$ in the multiplicative group $\mathbb{Z}_p^\ast$. We have

\begin{lemma}Let $p$ be an odd prime with $p\equiv1\pmod{4}$, and $B$  a perfect $B[-1, 3]^\ast(p)$ set for  $p$. Then $o(-2/3)$ in $\mathbb{Z}_p^\ast$ is odd.
\end{lemma}
\begin{proof}  If $o(-2/3)$ is even, then $-1\in<-\frac{3}{2}>$. By Theorem \ref{thn1}, we have $-i\in B$ when $i\in B$, which is impossible. Hence $o(-2/3)$ in $\mathbb{Z}_p^\ast$ is odd.
This proves the lemma.
\end{proof}

\begin{theorem}Let $p$ be an odd prime with $p\equiv5\pmod{8}$, then there exists a perfect $B[-1, 3]^\ast(p)$ set for  $p$ if and only if $6$ is a quartic residue modulo $p$. \end{theorem}
\begin{proof}If $p\equiv5\pmod{8}$ is a prime and there exists a perfect $B[-1, 3]^\ast(p)$ set for  $p$, let $g$ be a primitive root of modulo $p$,
$$N=\{ind_g(k)| k\in[-1, 3]^\ast\}, \quad A=\{ind_g(b) |b\in B\}.$$ By Lemma \ref{le28}, $N+A=\mathbb{Z}_{p-1}$ is a factorization. Since $\gcd(4, (p-1)/4)=1$, it follows from Theorem \ref{thn0} that $N$
is a complete set of residues modulo $4$, i.e., $N\pmod{4}=\{0, 1, 2, 3\}$. Note that $ind_g(1)\equiv0\pmod{4}$ and $ind_g(-1)=(p-1)/2\equiv2\pmod{4}$, so $\{ind_g(2)\pmod{3}, ind_g(3)\pmod{4}\}=\{1\pmod{4}, 3\pmod{4}\}$. It follows that $6=2\cdot3\equiv g^{1+4u}\cdot g^{3+4v}\equiv g^{4(u+v+1}\pmod{p}$, therefore
$$6^{\frac{p-1}{4}}\equiv1\pmod{p},$$
i.e., $6$ is a quartic residue modulo $p$.

If $6$ is a quartic residue modulo $p$, then it is easy to check that $N\pmod{4}=\{0, 1, 2, 3\}$. Let
$$A=\{0, 4, 8, \ldots, 4k, \ldots, p-5\},$$
then $N+A=\{0, 1, \ldots, p-2\}=\mathbb{Z}_{p-1}$ is a factorization of $\mathbb{Z}_{p-1}$, so $B=\{g^i\pmod{p}, i\in A\}$ is a perfect $B[-1, 3]^\ast(p)$ set for  $p$. This proves the theorem.\end{proof}

{\bf Remark:} By \cite{Su05} Corollary 5.2, for an odd prime $p$ with $p\equiv5\pmod{8}$,  $6$ is a quartic residue modulo $p$ if and only if $p=25x^2+14xy+25y^2$ or $p=5x^2\pm4xy+116y^2$.
Now by \cite{Co89} Theorem 9.12,  $5x^2\pm8xy+464y^2$ represents infinitely many prime numbers. Therefore there are infinitely many prime $p$ such that there exists a perfect $B[-1,3](p)$ set.

We now consider the case where $p$ is an odd prime with $p\equiv1\pmod{8}$. Let $g$ be a primitive root of modulo $p$, $p-1=2^tq, 2\not|q, t\ge3$. Then
$$2\equiv g^{2^ur}\pmod{p}, \quad 3\equiv g^{2^vs}\pmod{p}, \quad -1\equiv g^{2^{t-1}q}\pmod{p},$$
where $u, v, r, s$ are non-negative integers with $2^ur, 2^vs<p-1, 2\not|rs, u\ge1$. It is well-known that $<-1, 2, 3>=<g^{\gcd(2^ur, 2^vs, 2^{t-1}q)}>$ and
$$|<-1, 2, 3>|=\frac{p-1}{ \gcd(2^ur, 2^vs, 2^{t-1}q)}.$$
Since $ind_g(-\frac{3}{2})\equiv 2^vs-2^ur+2^{t-1}q\pmod{2^tq}$, so $o(-\frac{3}{2})$ is odd if and only if $2^t|ind_g(-\frac{3}{2})$, i.e. if and only if $2^vs-2^ur+2^{t-1}q\equiv0\pmod{2^t}$. If $2^{t-1}|\gcd(2^vs, 2^ur, 2^{t-1}q)$, then $4\not|\frac{p-1}{ \gcd(2^ur, 2^vs, 2^{t-1}q)}=|<-1, 2, 3>|$. Hence $[-1, 3]^\ast$ is not a direct factor of the subgroup generated by $\{-1, 2, 3\}$ since $|[-1, 3]^\ast|=4\not||<-1, 2, 3>|$. Therefore $[-1, 3]^\ast$ does not split $\mathbb{Z}_p$.

Now assume that $2^{t-1}\not|\gcd(2^vs, 2^ur, 2^{t-1}q)$, then $\min\{u, v\}<t-1$ and $u=v\le t-2$ since $2^vs-2^ur+2^{t-1}q\equiv0\pmod{2^t}$.

Finally, we assume that $u=v\le t-2$ and $2^vs-2^ur+2^{t-1}q\equiv0\pmod{2^t}$. Since
$$(-1)^{\frac{a-b}{2}}2^a3^b=6^{\frac{a+b}{2}}\cdot\left(-\frac{2}{3}\right)^{\frac{a-b}{2}}$$ when $a\equiv b\pmod{2}$ and
$$(-1)^{\frac{a-b-1}{2}}2^a3^b=2\cdot6^{\frac{a+b-1}{2}}\cdot\left(-\frac{2}{3}\right)^{\frac{a-b-1}{2}},$$
$$(-1)^{\frac{a-b+1}{2}}2^a3^b=3\cdot6^{\frac{a+b-1}{2}}\cdot\left(-\frac{2}{3}\right)^{\frac{a-b+1}{2}}$$when $a\equiv b+1\pmod{2}$.
We see that $[-1, 3]^\ast B$ is a factorization of $<-1, 2, 3>$, where $B=<6><-\frac{2}{3}>=\{2^a3^b\pmod{p}, a\equiv b\pmod{2}\}$ when $o(6)$ is odd, or $B=<6><-\frac{2}{3}>/\{-1, 1\}=\{2^a3^b\pmod{p}, a\equiv b\pmod{2}\}/\{-1, 1\}$ when $o(6)$ is even. From the above discussion we have proved the following Theorem.

\begin{theorem}Let $p$ be an odd prime with $p\equiv1\pmod{8}$, then there exists a perfect $B[-1, 3]^\ast(p)$ set for  $p$ if and only if $o(-\frac{3}{2})$ is odd and $4| o(2)$. \end{theorem}

\section{Conclusion} \label{conclusion}

In this paper, we  prove some new existence and nonexistence
results for perfect splitter sets.  For nonsingular
perfect splitter sets, we present new  necessary and sufficient conditions for prime $p$ such that there exists a nonsingular perfect $B[-1,3](p)$ set. We also show that there are infinitely many prime $p$ such that there exists a perfect $B[-1,3](p)$ set.  For purely singular perfect
splitter sets, we provide some general necessary conditions for
the existence of a purely singular perfect splitter set. As an
application, we determine all perfect $B[-k_1, k_2](2^n)$ sets for any positive integers $k_1,k_2$ with $k_1+k_2\ge4$.

\end{document}